\documentclass[runningheads]{llncs}
\usepackage{algpseudocode}
\usepackage{amsmath}
\usepackage{amsfonts}
\usepackage[shortlabels]{enumitem}
\usepackage[ruled,vlined,resetcount,linesnumbered]{algorithm2e}
\makeatletter
\newcommand{\nosemic}{\renewcommand{\@endalgocfline}{\relax}}% Drop semi-colon ;
\newcommand{\dosemic}{\renewcommand{\@endalgocfline}{\algocf@endline}}% Reinstate semi-colon ;
\makeatother

\usepackage{graphicx}
\begin{document}
\title{Improved Algorithms for Minimum-Membership Geometric Set Cover}
\author{Sathish Govindarajan \and
%Mayuresh Patle \and
Siddhartha Sarkar
}
\authorrunning{S. Govindarajan and S. Sarkar}
\institute{Indian Institute of Science, Bengaluru, India\\
\email{\{gsat, siddharthas1\}@iisc.ac.in}}
\maketitle              % typeset the header of the contribution
\begin{abstract}
Bandyapadhyay et al. introduced the generalized minimum-membership geometric set cover (GMMGSC) problem [SoCG, 2023], which is defined as follows.
We are given two sets $P$ and $P'$ of points in $\mathbb{R}^{2}$, $n=\max(|P|, |P'|)$, and a set $\mathcal{S}$ of $m$ axis-parallel unit squares. The goal is to find a subset $\mathcal{S}^{*}\subseteq \mathcal{S}$ that covers all the points in $P$ while minimizing $\mathsf{memb}(P', \mathcal{S}^{*})$, where $\mathsf{memb}(P', \mathcal{S}^{*})=\max_{p\in P'}|\{s\in \mathcal{S}^{*}: p\in s\}|$.

We study GMMGSC problem and give a $16$-approximation algorithm that runs in $O(m^2\log m + m^2n)$ time. Our result is a significant improvement to the $144$-approximation given by Bandyapadhyay et al. that runs in $\tilde{O}(nm)$ time. 

GMMGSC problem is a generalization of another well-studied problem called Minimum Ply Geometric Set Cover (MPGSC), in which the goal is to minimize the ply of $\mathcal{S}^{*}$, where the ply is the maximum cardinality of a subset of the unit squares that have a non-empty intersection. The best-known result for the MPGSC problem is an $8$-approximation algorithm by Durocher et al. that runs in $O(n + m^{8}k^{4}\log k + m^{8}\log m\log k)$ time, where $k$ is the optimal ply value [WALCOM, 2023].

\keywords{Computational Geometry \and Minimum-Membership Geometric Set Cover \and Minimum Ply Covering  \and Approximation Algorithms}
\end{abstract}
\section{Introduction}
\label{sec:intro}
Set Cover is a fundamental and well-studied problem in combinatorial optimization. Given a range space $(X, \mathcal{R})$ consisting of a set $X$ and a family $\mathcal{R}$ of subsets of $X$ called the ranges, the goal is to compute a minimum cardinality subset of $\mathcal{R}$ that covers all the elements of $X$. It is NP-hard to approximate the minimum set cover below a logarithmic factor \cite{feige_1998,Raz_setcover_hard_97}. When the ranges are derived from geometric objects, it is called the geometric set cover problem. Computing the minimum cardinality geometric set cover remains NP-hard even for simple 2D objects, such as unit squares on the plane \cite{geom_set_cover_2D_hardness}. There is a rich literature on designing approximation algorithms for various geometric set cover problems (see \cite{GeomSetCover_agarwal_pan_2014,Chan_Grant_pack_cover,Clarkson_Varadarajan_geomsetcover,Hochbaum_Maass_covering,raman_geom_setcover_focs2014}). Several variants of the geometric set cover problem such as unique cover, red-blue cover, etc. are well-studied \cite{CHAN_red_blue,takehiro_unique_coverage}.

In this paper, we study a natural variant of the geometric set cover called the Generalized Minimum-Membership Geometric Set Cover (GMMGSC). This is a generalization of two well-studied problems: minimum ply geometric set cover and minimum-membership geometric set cover, which were motivated by real-world applications in interference minimization in wireless networks and have received the attention of researchers \cite{socg23_minply,Biedl_Minply,Durocher_Mondal_Minply,erlebach_soda_minply}. We define the problem below. 

%More often than not, the geometric versions of the covering problems are efficiently solvable or approximated well. 
%Many variants of the Geometric Set Cover problem find applications in facility location, interference minimization in wireless networks, VLSI design, etc \cite{Demaine_application_of_coverage, geomsetcover_wireless}. 

\begin{definition}[Membership]
Given a set $P$ of points and a set $\mathcal{S}$ of geometric objects, the \textit{membership} of $P$ with respect to $\mathcal{S}$,  denoted by $\mathsf{memb}(P, \mathcal{S})$, is $\max_{p\in P}|\{s\in \mathcal{S}: p\in s\}|$.
% $\mathsf{ply}(\mathcal{S})$, is the maximum cardinality of any subset of $\mathcal{S}$ that has a common intersection.  
\end{definition}

\begin{definition}[GMMGSC problem]
Given two sets $P$ and $P'$ of points in $\mathbb{R}^{2}$, $n=\max(|P|, |P'|)$, and a set $\mathcal{S}$ of $m$ axis-parallel unit squares, the goal is to find a subset $\mathcal{S}^{*}\subseteq \mathcal{S}$ that covers all the points in $P$ while minimizing $\mathsf{memb}(P', \mathcal{S}^{*})$.
\end{definition}

\subsection{Related Work}
\label{subsec:related_work}
Bandyapadhyay et al. introduced the generalized minimum-membership geometric set cover (GMMGSC) problem and gave a polynomial-time constant-approximation algorithm for unit squares \cite{socg23_minply}. Specifically, they consider the special case when the points in $P$ lie within a unit grid cell and all the input unit squares intersect the grid cell. They use linear programming techniques to obtain a $16$-approximation in $\tilde{O}(nm)$ time for GMMGSC problem for this special case. Here, $\tilde{O}(\cdot)$ hides some polylogarithmic factors. This implies a $144$-approximation for GMMGSC problem for unit squares.  

We note that GMMGSC problem is a generalization of two well-studied problems: (1) Minimum-Membership Geometric Set Cover problem where $P'=P$, and (2) Minimum Ply Geometric Set Cover problem where $P'$ is obtained by picking a point from each distinct \textit{region} in the \textit{arrangement} $\mathcal{A}(\mathcal{S})$ of $\mathcal{S}$.

Minimum-Membership Set Cover (MMSC) problem is well-studied in both abstract \cite{kuhn_mmsc} and geometric settings \cite{erlebach_soda_minply}. 
Kuhn et al. showed that the abstract MMSC problem admits an $O(\log m)$-approximation algorithm, where $m$ is the number of ranges. They also showed that, unless P$=$NP, this is the best possible approximation ratio.
Erlebach and van Leeuwen introduced the geometric version of the MMSC problem \cite{erlebach_soda_minply}. They showed  NP-hardness for approximating the problem with ratio less than $2$ on unit disks (i.e., disks with diameter $1$) and unit squares. They gave a $5$-approximation algorithm for unit squares that runs in $n^{O(k)}$ time, where $k$ is the minimum membership.
 
Biedl et al. introduced the Minimum Ply Geometric Set Cover (MPGSC) problem \cite{Biedl_Minply}. They gave $2$-approximation algorithms for unit squares and unit disks that run in $nm^{O(k)}$ time, where $k$ is the optimal ply of the input instance. Durocher et al. presented the first constant approximation algorithm for MPGSC problem with unit squares \cite{Durocher_Mondal_Minply}. 
 They divide the problem into subproblems by using a standard grid decomposition technique. They solve almost optimally the subproblem within a square grid cell using a dynamic programming scheme. 
 Specifically, they give an algorithm that runs in $O(n + m^{8}k^{4}\log k + m^{8}\log m\log k)$ time and outputs a solution with ply at most $8k + 32$, where $k$ is the optimal ply.
Bandyapadhyay et al. also gave a $(36+\epsilon)$-approximation algorithm for MPGSC problem with unit squares that runs in $n^{O(1/\epsilon^{2})}$ time \cite{socg23_minply,Mustafa_Ray}. 

%To the best of our knowledge, the Minimum Ply Cover Problem has not been studied in the continuous setting. The standard geometric set cover problem has been well studied in the continuous setting. Fowler et al. showed that the standard set cover for unit squares in the continuous setting is NP-complete \cite{geom_set_cover_2D_hardness}. Hochbaum et al. and Gonzalez presented polynomial-time approximation schemes (PTAS) for covering $n$ points in $d$-dimensional space with minimum number of $d$-dimensional balls and fixed-size orthogonal hypersquares, respectively \cite{Hochbaum_Maass_covering}\cite{gonzalez}. There have been many efficient constant factor approximations for this problem \cite{almost-linear-disc-cover, biniaz-udc-2d-3d, experiments-with-disk-7-approx}. todo

%Okayama et al. \cite{okayama-disjoint-disk} and Aloupis et al. \cite{aloupis-disjoint-disk} studied the Disjoint Unit Disk Cover problem, which involves covering a given set of points on a plane using non-intersecting disks that may touch. They provided bounds for size $k$ of the smallest point set that cannot be covered by disjoint unit disks, the latest proved bounds being $13 \leq k \leq 50$. Nan Hu \cite{nan-hu-disjoint-disk-np} proved that it is NP-complete to determine if a given set of points can be covered using disjoint unit disks. Note that the Disjoint Unit Disk Cover Problem is essentially  asking if a 1-ply cover solution exists in the continuous setting.

\subsection{Our Contribution} \label{subsec:our_contri}
We first consider a special case of the GMMGSC problem called the line instance of GMMGSC, where the input squares are intersected by a horizontal line and the points to be covered lie on only one side of the line. Refer to Definition \ref{defn:minply_line}. We design a polynomial-time algorithm (i.e., Algorithm \ref{algo:exact_algo_hline}) for this problem where the solution has some desirable properties. 

Next, we consider the slab instance of GMMGSC, where the input squares intersect with a unit-height horizontal slab and the points to be covered lie within the slab. Refer to Definition (\ref{defn:minply_slab}). As far as we know, there are no known approximation results for this problem. We adapt the linear programming techniques in \cite{socg23_minply} to decompose a slab instance of GMMGSC into two line instances of GMMGSC. We use Algorithm \ref{algo:exact_algo_hline} to solve them. Then we merge the two solutions to obtain the final solution. A major challenge was finding a solution for the line GMMGSC which respects a key lemma (i.e., Lemma \ref{lemma:excl}). This key lemma enables us to obtain a solution with membership at most $(8\cdot OPT+18)$ for the slab instance, where $OPT$ is the minimum membership. 

Finally, we give an algorithm for GMMGSC problem for unit squares that runs in $O(m^{2}\log m + m^{2}n)$ time and outputs a solution whose membership is at most $16\cdot OPT+36$. We divide GMMGSC instance into multiple line instances. Then we use Algorithm \ref{algo:exact_algo_hline} on the line instances. Finally, we merge the solutions of the line instance to obtain the final solution. 

For GMMGSC problem, we note that our result is a significant improvement in the approximation ratio as compared to the best-known result of Bandyapadhyay et al \cite{socg23_minply}. For MPGSC problem, our result is a significant improvement in the running time as compared to the best-known result of Durocher et al. while achieving a slightly worse approximation ratio \cite{Durocher_Mondal_Minply}.
% To the best of our knowledge, this is the first study of the Minimum Ply Cover problem in the continuous setting. We present a simple $O(n\log n)$ time algorithm to find a $1$-ply cover for any given set of $n$ points in the plane using axis-aligned unit squares. Note that a $1$-ply cover is the best possible solution for the Minimum Ply Cover problem. It is quite surprising that the Minimum Ply Cover problem in the Continuous Setting for unit squares admits a simple optimal ($1$-ply cover) solution, whereas the corresponding standard geometric set cover problem is NP-complete. We further show that the size of our 1-ply cover is at most twice that of a minimum-sized 1-ply cover. We also generalize this algorithm for covering points in $d$-dimensional space using fixed-sized axis-aligned hyperboxes, achieving $1$-ply cover with a $O(dn \log n)$ time complexity.  We show that the size of our 1-ply cover is at most $2^{d-1}$ times the size of a minimum-sized 1-ply cover. Finally, we consider the Minimum Ply Cover Problem using convex polygons with $m$ vertices and provide a 4-ply cover that can be computed in $O(n\log n+mn)$ time.

\section{Generalized Minimum-Membership Set Cover for Unit Squares} \label{sec:discrete}
Let $P$ and $P'$ be two sets of points in $\mathbb{R}^{2}$ and $\mathcal{S}$ be a set of axis-parallel unit squares. We want to approximate the minimum-membership set cover (abbr. MMSC) of $P$ using $\mathcal{S}$ where membership is defined with respect to $P'$. First, we divide the plane into horizontal slabs of unit height. Each slab is defined by two horizontal lines $L_1$ and $L_2$, unit distance apart, where $L_2$ is above $L_1$. We define an instance for the slab subproblem below. For an illustration, refer to Figure \ref{fig:line_slab_instances}.

\begin{definition}[Slab instance]\label{defn:minply_slab}
Consider a set $\mathcal{S}$ of unit squares where each square intersects one of the boundaries of a unit-height horizontal slab $\alpha$. The points of the set $P$ to be covered are located within $\alpha$, each point lying inside at least one of the squares in $\mathcal{S}$. Let $P'$ be a set of points with respect to which the membership is to be computed. The instance $(P, P', \mathcal{S})$ is called a slab instance.
\end{definition}

 \begin{figure}[ht!]
    \centering
    \includegraphics[width=12cm]{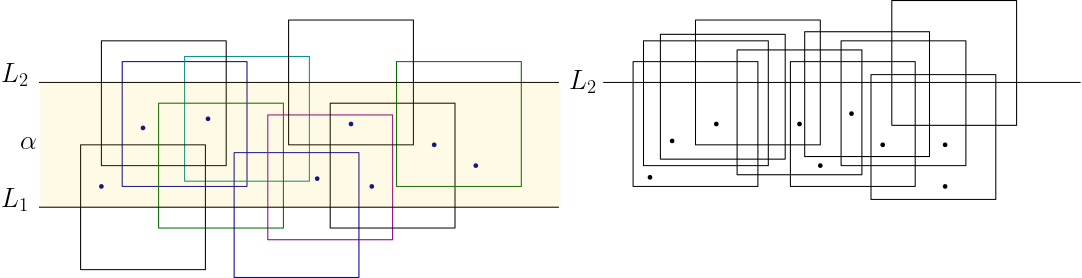}
    \caption{A slab instance and a line instance of GMMGSC.}
    \label{fig:line_slab_instances}
\end{figure}

%Given a slab instance, we divide it into two line instances of MPC as defined in \ref{defn:minply_line}. One has $L_2$, i.e., the top boundary line of $\alpha$ as the intersecting horizontal line, and the input points lie below $L_2$. The other instance has $L_1$, i.e., the bottom boundary line of $\alpha$ as the intersecting horizontal line, and the input points lie above $L_1$. We will solve the two instances and return the union of the two solutions. In the next section, we formally define and solve the line instance.
In section \ref{sec:slab_instance}, we solve the slab instance by decomposing it into two \textit{line instances}. In the following section, we define and discuss the line instance.
\subsection{GMMGSC for the Line Instance}\label{sec:line_instance}

\begin{definition}[Line instance]\label{defn:minply_line}
Consider a set $\mathcal{S}$ of unit squares where each square intersects a horizontal line $\ell$. The points of the set $P$ to be covered are located only on one side (i.e., above or below) of $\ell$, each point lying inside at least one of the squares in $\mathcal{S}$. Let $P'$ be a set of points with respect to which the membership is to be computed. The instance $(P, P', \mathcal{S})$ is called a line instance.
\end{definition}

In the rest of this section, we design an algorithm for the line instance where the points in $P$ lie below the defining horizontal line. For the slab, it would be the instance corresponding to the top boundary line $L_2$ of the slab. Refer to Figure \ref{fig:line_slab_instances} for an example. The algorithm for the line instance corresponding to the bottom boundary line $L_1$ is symmetric. 

Let us introduce some notation first. For a unit square $s\in \mathcal{S}$, denote by $x(s)$ and $y(s)$ the $x$-coordinate and $y$-coordinate of the bottom-left corner of $s$, respectively. For a horizontal line $\ell$, denote by $y(\ell)$ the $y$-coordinate of any point on $\ell$.

We make the following non-degeneracy assumptions. First, no input square has its top boundary coinciding with the slab boundary lines. Second, $x$- and $y$-coordinates of the input squares are distinct. Note that a set $Q$ of intersecting unit squares also forms a clique in the intersection graph of $Q$, and vice versa. 

In a set of unit squares $\mathcal{S}$, two squares $s_1, s_2 \in \mathcal{S}$ are \textbf{consecutive} (from left to right) if there exists no square $t\in \mathcal{S}$ such that $x(s_1)< x(t) < x(s_2)$. 
If a point $p$ is contained in exactly one square $s$ in a set cover $\mathcal{S}^{*}$ of a set of points $P$, then $p$ is said to be an \textbf{exclusive point} of the square $s$ with respect to $\mathcal{S}^{*}$. 
For $s\in \mathcal{S}^{*}$, the region in the plane, denoted by $\mathsf{Excl}(s)$, which is covered exclusively by $s$ is called the \textbf{exclusive region} of $s$ with respect to a set cover $\mathcal{S}^{*}$. 
For $s_i, s_j\in \mathcal{S}^{*}$, the region in the plane, denoted by $\mathsf{Excl}(s_i, s_j)$, which is contained exclusively in $s_i \cap s_j$, is called the \textbf{pairwise exclusive region} of $s_i$ and $s_j$ with respect to a set cover $\mathcal{S}^{*}$. 
A square $s$ in a set cover $\mathcal{S}^{*}$ of a set of points $P$ is called \textbf{redundant} if it covers no point of $P$ exclusively. Refer to Figure \ref{fig:caldem1} for an illustration of these terms. 

\begin{figure}[ht!]
\centering
\includegraphics[width=6cm]{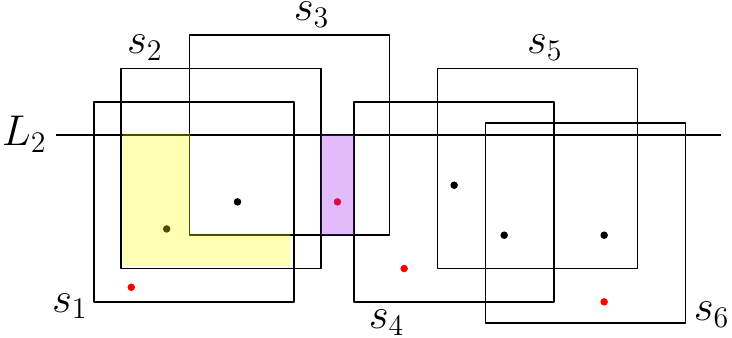}
\caption{A set cover for a line instance. For $i\in [5]$, $s_i, s_{i+1}$ are consecutive squares. The red points are exclusive points. The purple region is $\mathsf{Excl}(s_3)$. The yellow region is $\mathsf{Excl}(s_1, s_2)$. The squares $s_2, s_5$ are redundant. }
\label{fig:caldem1}
\end{figure}

A set of unit squares having a common intersection is said to form a \textbf{geometric clique}.
A set of unit squares containing a point of $P'$ in their common intersection region is said to form a \textbf{discrete clique}.
The common intersection region of a set of unit squares forming a clique $Q$ is called the \textbf{ply region} of $Q$. 
The ply region of a clique $Q$ is always rectangular. For a clique $Q$, denote by $x_{l}(Q)$ (resp. $x_{r}(Q)$) the $x$-coordinate of the left (resp. right) boundary of the ply region of $Q$. Unless specified otherwise, a clique refers to a discrete clique.
%------------------------------------------------------
\subsubsection{Types of legal cliques}
First, we classify a clique with respect to a line instance.
A set of intersecting squares in a line instance is called a \textbf{top-anchored clique} when the points to be covered lie below the line with respect to which the line instance is defined.  
A set of intersecting squares in a line instance is called a \textbf{bottom-anchored clique} when the points to be covered lie above the line with respect to which the line instance is defined.  

Let a line instance be defined with respect to a horizontal line $\ell$. Let $s_1, \ldots, s_k$ be a sequence of squares from left to right having a common intersection. This set of squares is called a \textbf{monotonic ascending clique} if $i < j$ implies $y(s_i)< y(s_j)$, for all $i,j$. We use the abbreviation $ASC$ to denote such a clique. On the other hand, if $i < j$ implies $y(s_i)> y(s_j)$,  for all $i,j$, then this set of squares is called a \textbf{monotonic descending clique}. We use the abbreviation $DESC$ to denote such a clique.

Let a line instance be defined with respect to a horizontal line $l$. Let $s_1, \ldots, s_k,$
$s_{k+1}, \ldots s_{k+r}$ be a sequence of squares from left to right having a common intersection. This set of squares is called a \textbf{composite clique} if the following holds.
    \begin{itemize}
        \item The sequence of squares $s_1, \ldots, s_k$ forms a monotonic clique.
        \item Either $y(s_{k+1})>y(s_{k})< y(s_{k-1})$, 
        or $y(s_{k-1})< y(s_{k})> y(s_{k+1})$, and
        \item The sequence of squares $s_{k+1}, \ldots, s_{k+r}$ forms a monotonic clique.
    \end{itemize}
    The square $s_k$ is called the \textit{transition square}. We use the abbreviation $DESC|ASC$ to denote a composite clique where the sequence $s_1, \ldots, s_k$ is descending but the sequence $s_{k+1},\ldots, s_{k+r}$ is ascending. For other types of composite cliques, the abbreviation would be self-explanatory.

\begin{figure}[ht!]
\centering
\includegraphics[width=10cm]{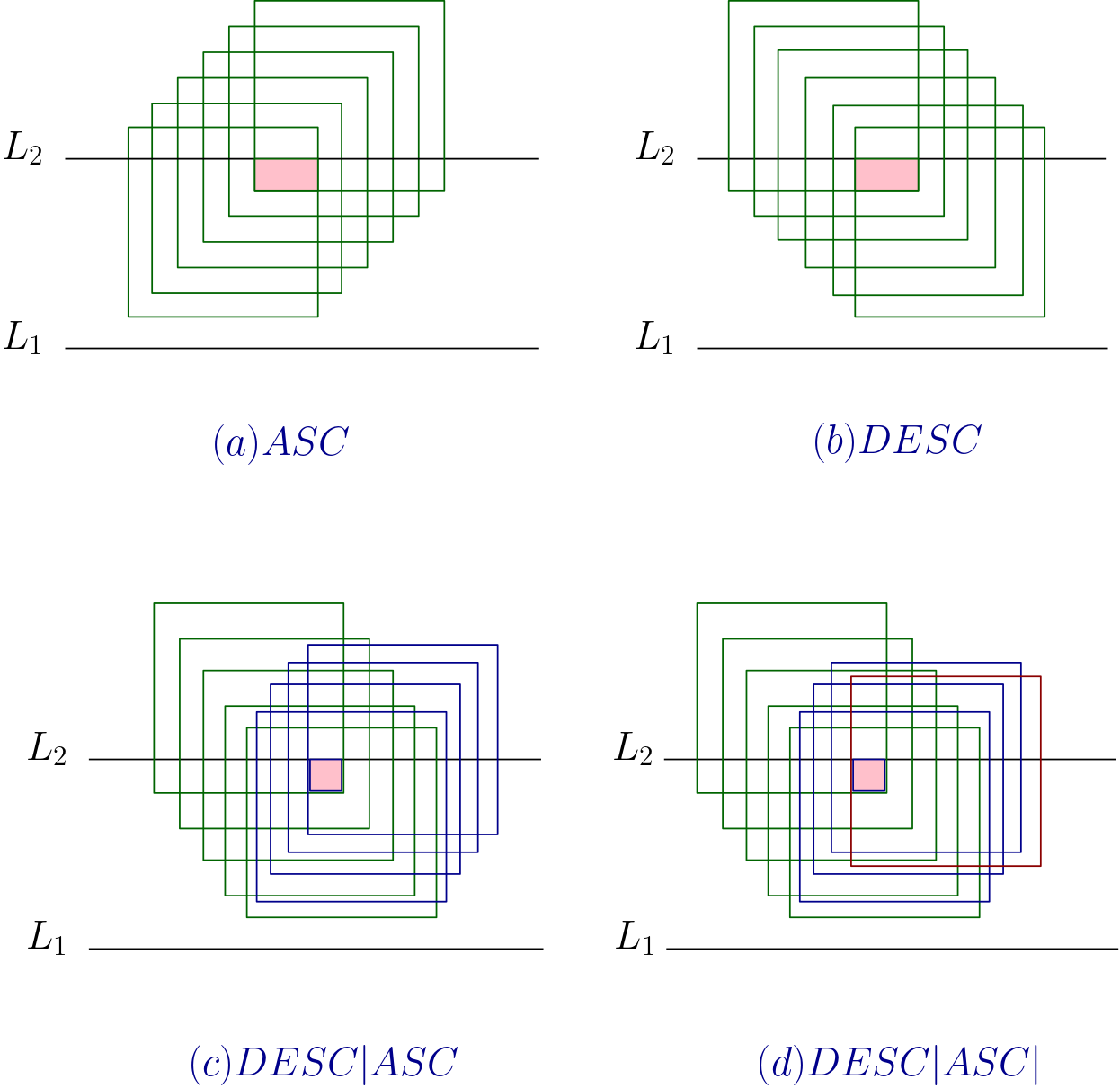}
\caption{(a), (b), and (c) show three different types of legal top-anchored cliques. (d) Shows an invalid clique where a top-anchored composite $DESC|ASC$ clique is followed by another \textit{transition}.}
\label{fig:top_all}
\end{figure}

\begin{claim}\label{claim:forbidden_1}
    In a set cover for the line instance, where none of the constituent squares are redundant, a top-anchored composite clique must be of type $DESC|ASC$.
\end{claim}
\begin{proof}
    Suppose not. Let $Q$ be a top-anchored composite clique of type $ASC|ASC$, where the left ascending sequence be the squares $s_1, \ldots, s_{k}$ from left to right. Then the square $s_k$ would become redundant since the two squares $s_{k-1}$, $s_{k+1}$ would cover all the points covered by $s_k$ lying below the defining horizontal line. This contradicts the non-redundancy condition. For an example, refer to Figure (\ref{fig:top_asc_asc}). 
    Similar arguments apply to rule out the existence of top-anchored composite cliques of type $ASC|DESC$ and $DESC|DESC$.
\end{proof}
In the rest of the paper, whenever we consider a top-anchored clique in the solution (containing no redundant squares) of a line instance of the GMMGSC problem, we assume that the clique is of one of the following legal types: (i) monotonic $ASC$, (ii) monotonic $DESC$, or (iii) composite $DESC|ASC$.
\begin{figure}[ht!]
\centering
\includegraphics[width=13cm]{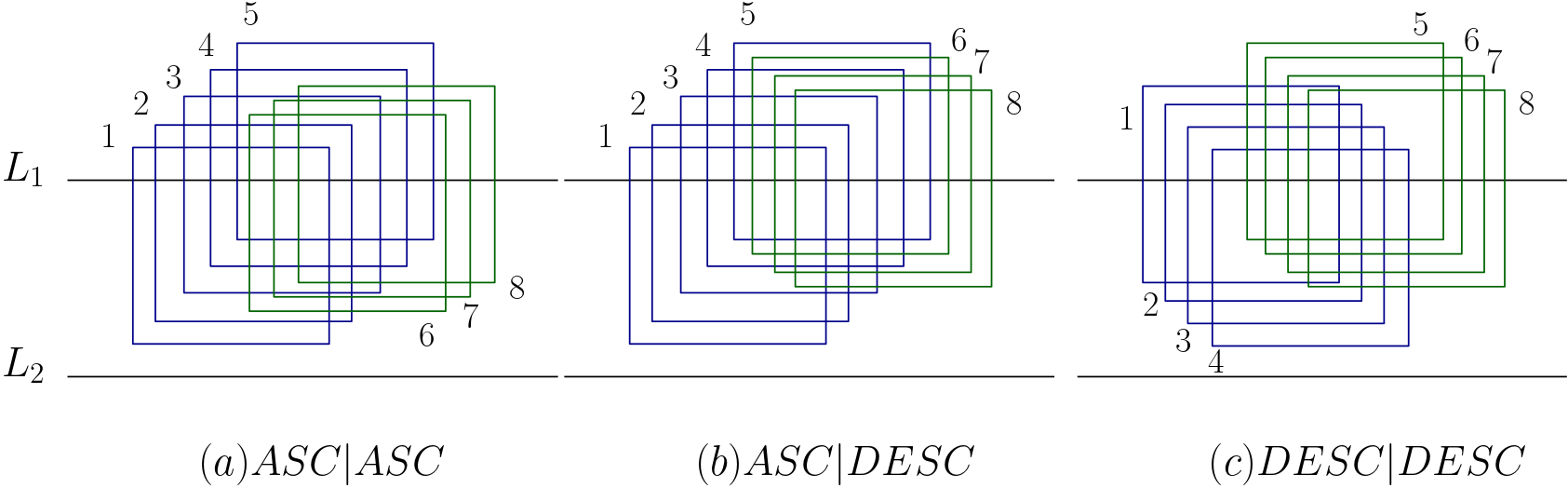}
\caption{Different types of a forbidden top-anchored clique. In all the figures, the square $5$ is redundant since the squares $4$ and $6$ fully cover the area of square $5$ below the line $L_1$.}
\label{fig:top_asc_asc}
\end{figure}
%--------------------------------------------------------------
\subsubsection{The Algorithm}
In this subsection, we describe an algorithm for the line instance $(P, P', \mathcal{S})$ that produces a feasible set cover $\mathcal{S}^{*}$ with some desirable structural properties. Multiple maximum cliques may exist in the intersection graph of $\mathcal{S}^{*}$. We order the maximum cliques of $\mathcal{S}^{*}$ in the increasing order of the $x_{r}(\cdot)$ values of their ply regions. 

\begin{definition}[Leftmost maximum clique]\label{defn: l_max_clique}
The leftmost maximum clique $Q$ of a set of unit squares $\mathcal{S}^{*}$ refers to that maximum clique in the intersection graph of $\mathcal{S}^{*}$ for which $x_{r}(\cdot)$ value of the corresponding ply region is the minimum among all the maximum cliques in $\mathcal{S}^{*}$.
\end{definition}
The procedure $\mathsf{Swap}(\mathcal{A}, \{s\}, \mathcal{S}^{*})$ consists of the deletion of a set of squares $\mathcal{A}$ from a set cover $\mathcal{S}^{*}$ and the addition of a square $s\in \mathcal{S}\setminus \mathcal{S}^{*}$ into $\mathcal{S}^{*}$.

\begin{definition}[Profitable swap]\label{defn:profit_swap}
The operation $\mathsf{Swap}(\mathcal{A}, \{s\},  \mathcal{S}^{*})$ is a profitable swap if $\mathcal{A}$ is a set of two or more consecutive squares in the leftmost maximum clique $Q\subseteq \mathcal{S}^{*}$, and $s\in \mathcal{S}\setminus \mathcal{S}^{*}$ such that $\mathcal{S}^{*}\cup \{s\}\setminus \mathcal{A}$ is a feasible set cover for $P$.
%\item and the size of the clique $Q$ decreases.
\end{definition}

The procedure $\mathsf{RemoveRedundancy}(\mathcal{S})$ ensures that each square $s\in \mathcal{S}^{*}$ contains at least one point $p\in P$ exclusively. We proved that in a set cover for the line instance, where none of the constituent squares are redundant, any maximum clique can be of only $3$ types, as shown in Figure \ref{fig:domination_lemma}.

\begin{algorithm}[ht]
\SetAlgoLined
\SetKwInOut{Input}{Input}
\SetKwInOut{Output}{Output}
\Input{ A horizontal line $\ell$, a set $\mathcal{S}$ of $m$ unit squares intersecting $\ell$, a set $P$ of $n$ points below $\ell$ such that each of them lies in at least one square in $\mathcal{S}$, and a set of points $P'$.}
\Output{ Returns a set of squares $\mathcal{S}^{*} \subseteq \mathcal{S}$ covering $P$.}
\nosemic
$ \mathcal{S}^{*} \leftarrow \mathsf{RemoveRedundancy}(\mathcal{S})$\;
Let $Q$ be the leftmost maximum clique in $\mathcal{S}^{*}$.\;
\While{\text{there is a $\mathsf{profitable\: swap}$ in $Q$}}{
$\mathcal{S}^{*} \leftarrow \mathsf{Swap}(\mathcal{A}, \{s\}, \mathcal{S}^{*})$\;
$ \mathcal{S}^{*} \leftarrow \mathsf{RemoveRedundancy}(\mathcal{S}^{*})$\;
Update $Q$ to be the leftmost maximum clique in $\mathcal{S}^{*}$.\;
}
\Return $\mathcal{S}^{*}$
\caption{Algorithm for line instance of GMMGSC}
\label{algo:exact_algo_hline}
\end{algorithm}

The algorithm for the line instance is given in Algorithm \ref{algo:exact_algo_hline} and has two steps. In the first step, $\mathsf{RemoveRedundancy(\cdot)}$ is applied on the set $\mathcal{S}$ of input squares to obtain a feasible set cover $\mathcal{S}^{*}$ containing no redundant squares. The second step performs a set of \textit{profitable swap}s on $\mathcal{S}^{*}$. This step aims to obtain a feasible solution with a maximum clique $Q$ with some desirable properties. A profitable swap is defined on the leftmost maximum clique of a feasible solution. Refer to Definition \ref{defn:profit_swap}. We have the following observation about the solution.\\
\begin{observation}\label{obs:1}
There exist no profitable swaps on the leftmost maximum clique $Q$ of $\mathcal{S}^{*}$ returned by Algorithm \ref{algo:exact_algo_hline}.    
\end{observation}

Our algorithm implicitly implies the following lemma.
\begin{lemma}[Key Lemma]\label{lemma:excl}
Let $Q$ be the leftmost maximum clique in the solution $\mathcal{S}^{*}$ returned by Algorithm \ref{algo:exact_algo_hline}. Let $|Q|=k$ and $s_1, \ldots, s_k$ be the squares of $Q$ from left to right. No input square contains $\mathsf{Excl}(s_i)\cup \mathsf{Excl}(s_{i+1})$, where $1\leq i\leq k-1$.
\end{lemma}
\begin{proof}
Fix an index $i$ with $1\leq i\leq k-1$. Suppose for the sake of contradiction that there exists an input square $t$ that covers $\mathsf{Excl}(s_i)\cup \mathsf{Excl}(s_{i+1})$. Observe that $\mathsf{Swap}(\{s_i, s_{i+1}\}, \{t\}, \mathcal{S}^{*})$ is a profitable swap in $\mathcal{S}^{*}$, which is a contradiction to Observation \ref{obs:1}.
\end{proof}

First, we will show how to implement the algorithm and analyze the running time. 

\begin{observation}\label{obs:2}
The procedure $\mathsf{RemoveRedundancy(\mathcal{S})}$ can be implemented in $O(nm)$ time.   
\end{observation}
\begin{proof}
Consider an arbitrary ordering of the $n$ points in $P$ and an arbitrary ordering of the $m$ squares in $\mathcal{S}$. In $O(nm)$ time, we construct an $n\times m$ matrix $T$ where the $(i, j)$-th entry is $1$ if the $i$-th point is contained in the $j$-th square. Also, construct an array $\mathsf{count}$ of length $n$. For every point $p_{i}$, $\mathsf{count}[i]$ stores the number of squares in $\mathcal{S}$ that contain $p_{i}$. Initialize each entry of an $m$-length array $\mathsf{removed}$ to $0$. Run a loop that iterates over each square $s_j\in \mathcal{S}$. If $s_j$ does not contain any point $p_{i}$ with $\mathsf{count}[i]=1$, remove $s_j$, i.e., set $\mathsf{removed}[j]=1$. Set each entry in the $j$-th column of $T$ to zero. For each point $p_{i'}\in s_j$, decrease $\mathsf{count}[i']$ by one. At the end of the loop, output the squares with $\mathsf{removed}[\cdot] = 0$. Naively, the time required to perform all the operations is $O(nm)$.
\end{proof}

\begin{theorem}\label{thm:line_run_time}
The running time of Algorithm \ref{algo:exact_algo_hline} is $O(m^{2}\log m + m^{2}n)$.
\end{theorem}
\begin{proof}
To compute the leftmost maximum clique $Q$ in $\mathcal{S}^{*}$, we obtain the set of squares containing each point of $P'$. This takes $O(nm)$ time.

We need to check if there exists a profitable swap in $Q$. There are at most $m$ choices for the swapped-in square $s\in \mathcal{S}\setminus\mathcal{S}^{*}$. While computing $Q$, one could also obtain the left-to-right ordering of the squares in $Q$. For each candidate swapped-in square $s$, we find a set of consecutive squares $\mathcal{A} \subseteq Q$ that can lead to a profitable swap of $\mathcal{A}$ by $s$. We can use binary search on the squares of $Q$ to do this. This requires $O(\log |Q|) = O(\log m)$ time. Additionally, we may need to check if the extreme two squares of $\mathcal{A}$, say $s_i$ and $s_j$, can be swapped out safely. This is equivalent to checking if all points in $\mathsf{Excl}(s_i, s_j)$ are covered by $s$. Naively, the time required to check this is $O(n)$. Thus, if one exists, we can execute a profitable swap in $Q$ in $O(m\log m + mn)$ time. 

We need to remove from $\mathcal{S}^{*}$ those squares in $\mathcal{S}^{*} \setminus Q$ that may have become redundant because of swapping in the square $s$. This can be done by re-invoking $\mathsf{RemoveRedundancy(\cdot)}$ on $\mathcal{S}^{*}$ in $O(nm)$ time (due to Observation \ref{obs:2}). The leftmost maximum clique in $\mathcal{S}^{*}$ can be determined in $O(nm)$ time. 
%  If there exist \textit{profitable swaps} as defined below,  we perform the $\mathsf{swap}(\cdot)$ operation on $T[i, j]$ that results in the best solution as per the greedy criteria. Here $\mathcal{S}$ is the set of input squares.
 
Every profitable swap decreases the size of the set cover by at least one. Therefore, at most $m$ profitable swaps are performed. Thus, there are at most $m$ iterations of the \textit{while} loop. Hence, the total running time of the \textit{while} loop is $O(m\cdot(m\log m + mn)) = O(m^{2}\log m + m^{2}n)$. 
\end{proof}

% Removed lemma \label{lemma:plus_2_approx}
\subsubsection{Structural Properties of the solution}
We state two properties about the structure of the solution returned by Algorithm \ref{algo:exact_algo_hline}. 
%These properties are utilized to argue the approximation factor of our algorithm for the slab instance of MPC in Section \ref{sec:slab_instance}. 
Let $Q$ be the leftmost maximum clique in the solution $\mathcal{S}^{*}$ returned by Algorithm \ref{algo:exact_algo_hline} for the line instance $(P, P', \mathcal{S})$.  Let $s_1, \ldots, s_k$ be the squares of $Q$ from left to right. For $1\leq j\leq k$, let $p_j$ be the bottom-most exclusive point in $s_j$.
\begin{lemma}\label{lemma:domination}
Let $p\in P'$ be an arbitrary point contained in the common intersection region of $Q$. For $k> 13$, there exists a set $J\subset [k]$ with $|J|\geq k-9$ such that every input square $t\in \mathcal{S}$ containing $p_{j}$ also contains $p$, for $j\in J$.
\end{lemma}
\begin{proof}
There are two cases to consider.
\begin{enumerate}[start=1,label={\upshape \bfseries Case \arabic*:},wide = 0pt, leftmargin = 3em]
    \item $Q$ is a monotonic clique. There are two subcases.
    \begin{itemize}
        \item $Q$ is a monotonic descending clique. Define the set $J=\{3, \ldots, k-3\}$. By definition $y(s_{j-1})>y(s_{j})> y(s_{j+1})$, for each $j\in J$. Let $t$ be an input square that contains $p_j$ but not $p$. Again, there are two cases.
        \begin{itemize}
        \item $p$ lies to the right of $t$, i.e. $x(p)>x(t)+1$: Observe that $x(p_{j})< x(s_{k})$. Since $p\in s_1$ and $t$ does not contain $p$, therefore $t$ starts before $s_1$ starts, i.e., $x(t)<x(s_{1})$. Since $t$ contains $p_{j}\in s_{j}$, hence $t$ must end below the bottom boundary of $s_{j-1}$ and should end to the right of where $s_j$ starts, i.e., $y(t)< y(s_{j-1})$ and $x(t)+1> x(s_{j})$. Combining these with the fact that $t$ is a unit square, $t$ must cover $\mathsf{Excl}(s_{j-2})\cup \mathsf{Excl}(s_{j-1})$ as shown in Figure \ref{fig:domination_lemma}(b). This is a violation of Lemma \ref{lemma:excl}.
        
        \item $p$ does not lie to the right of $t$ but is above $t$, i.e., $x(t)+1>x(p)$ but $y(p)> y(t)+1$: The square $t$ must end after $s_k$ starts, i.e., $x(t)+1>x(s_k)$. Since $p\in s_k$, the square $t$ must end below $s_k$, i.e., $y(t)<y(s_k)$. Since $p_j\in t$, therefore $x(t)<x(p_j)$. So, $t$ covers $\mathsf{Excl}(s_{j+1})\cup \mathsf{Excl}(s_{j+2})$. This is a violation of Lemma \ref{lemma:excl}.
    \end{itemize}
        \item $Q$ is a monotonic ascending clique. Define the set $J=\{4, \ldots, k-2\}$. By definition $y(s_{j-1})<y(s_{j})<y(s_{j+1})$, for each $j\in J$. Let $t$ be an input square that contains $p_j$ but not $p$. By an analogous argument, $t$ would be forced to cover either $\mathsf{Excl}(s_{j+1})\cup \mathsf{Excl}(s_{j+2})$ (when $p$ lies to the left of $t$ as shown in Figure \ref{fig:domination_lemma}(a)) or $\mathsf{Excl}(s_{j-1})\cup \mathsf{Excl}(s_{j-2})$ (when $p$ does not lie to the left of $t$ but lies above $t$). This is again a violation of Lemma \ref{lemma:excl}.
    \end{itemize}
    \item $Q$ is a composite clique (of type $DESC|ASC$). Let $b$ be the index (in the left-to-right ordering) of the bottom-most square in $Q$. Define the set $J=\{3, \ldots, b-3, b+3, \ldots, k-2\}$. Using the two subcases in the previous case, we can arrive at a violation of Lemma $1$ for $j\in J$. Refer to Figure \ref{fig:domination_lemma}(c) for an illustration.
\end{enumerate}
\end{proof}

\begin{lemma}\label{lemma:domination2}
For $k> 13$, there exists a set $J\subset [k]$ with $|J|\geq k-5$ such that no input square $t\in \mathcal{S}$ can contain $p_{j}, p_{j+1}, p_{j+2}$ for $j\in J$.
\end{lemma}
\begin{proof}
There are two cases to consider. 
\begin{enumerate}[start=1,label={\upshape \bfseries Case \arabic*:},wide = 0pt, leftmargin = 3em]
    \item $Q$ is a monotonic clique. There are two subcases.
    \begin{itemize}
        \item $Q$ is a monotonic descending clique. Define the set $J=\{1, \ldots, k-3\}$. 
        \item $Q$ is a monotonic ascending clique. Define the set $J=\{2, \ldots, k-2\}$.
    \end{itemize}
    For $j \in J$, assume that $t\in \mathcal{S}$ contains the bottom-most exclusive points $p_{j}, p_{j+1}, p_{j+2}$ of three consecutive squares of $Q$, namely $s_{j}, s_{j+1}, s_{j+2}$ respectively. Then $t$ would contain either $\mathsf{Excl}(s_{j})\cup \mathsf{Excl}(s_{j+1})$ or $\mathsf{Excl}(s_{j+1})\cup \mathsf{Excl}(s_{j+2})$. This implies a violation of Lemma \ref{lemma:excl}, as can be seen in Figure \ref{fig:domination_lemma}. Thus, we have a contradiction.
    \item $Q$ is a composite clique (of type $DESC|ASC$). Let $b$ be the index (in the left-to-right ordering) of the bottom-most square in $Q$. We define the set $J = \{1, \ldots, b-3, b+1, \ldots, k-2\}$. If $1\leq j\leq b-3$, then $s_{j}, s_{j+1}, s_{j+2}$ are in a monotonic descending sequence and the corresponding subcase from Case $1$ applies. If $b+1 \leq j\leq k-2$, then $s_{j}, s_{j+1}, s_{j+2}$ are in a monotonic ascending sequence and the corresponding subcase from Case $1$ applies.
\end{enumerate}
\end{proof}
\begin{figure}[ht!]
    \centering
    \includegraphics[width=11cm]{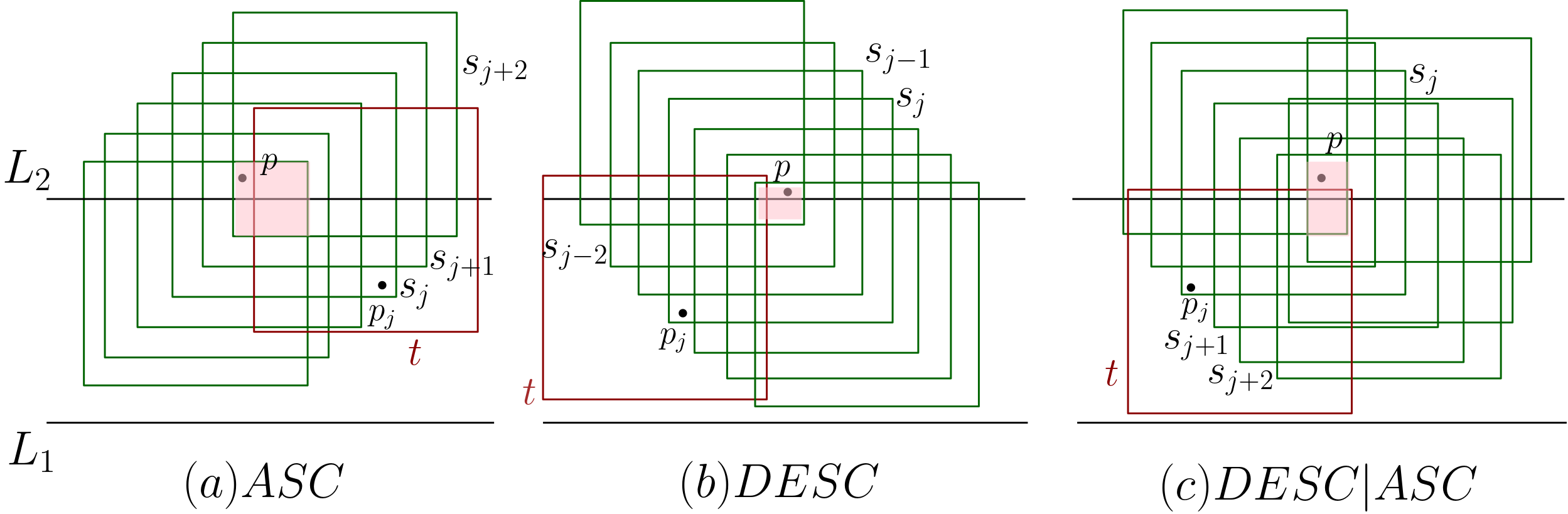}
    \caption{$p\in P'$ is contained in the ply region of the cliques. The red square $t$ does not contain $p$. (a) The green squares constitute a monotonic $ASC$ clique. (b) The green squares constitute a monotonic $DESC$ clique. In (c), the green squares form a composite clique of type $DESC|ASC$. In (a),the red square $t$ can swap out $s_{j+1}, s_{j+2}$. In (b),the red square $t$ can swap out $s_{j-1}, s_{j-2}$.}
    \label{fig:domination_lemma}
\end{figure}

\subsection{GMMGSC for the Slab Instance}\label{sec:slab_instance}
In this section, we present a constant approximation algorithm for the slab instance. We will use an LP relaxation (adapted from Bandyapadhyay et al. \cite{socg23_minply}) to partition the slab instance into two line instances.

For each unit square $s \in \mathcal{S}$, we create a variable $x_{s}$, indicating whether $s$ is included in our solution. In addition, create another variable $y$, which indicates the maximum number of times a point in $P'$ is covered by our solution. %For the set $\mathcal{S}$ of unit squares, define $\mathcal{C}$ as the set of regions in the arrangement induced by $\mathcal{S}$. Pick a point from each region in $\mathcal{C}$ and form the set $P'$ of points. 
Then, we formulate the following linear programming relaxation.

\begin{equation*}
\begin{array}{rl}
    & \min \text{ } y  \\[2ex]
    \text{s.t.} & \sum\limits_{s\in \mathcal{S}, p\in s}  x_{s} \geq 1\quad \text{for all } p \in P\\[1ex]
    &  \sum\limits_{s\in \mathcal{S}, p'\in s}  x_{s} \leq y\quad \text{for all } p' \in P'\\[1ex]
    & 0 \leq x_{s}\leq 1\quad  \text{for all } s \in \mathcal{S}\\
\end{array}
\end{equation*}

The input set $\mathcal{S}$ of squares is partitioned naturally into two parts $\mathcal{S}_1, \mathcal{S}_2$ where $\mathcal{S}_1$ (resp. $\mathcal{S}_2$) consists of the input squares intersecting the bottom (resp. top) boundary line $L_1$ (resp. $L_2$) of the horizontal slab, say $\alpha$.

We partition the set $P$ of points within $\alpha$ into $P_1$ and $P_2$ using the LP. Let $(\{x^{*}_{s}\}_{s\in\mathcal{S}}, y^{*})$ be an optimal solution of the above linear program computed using a polynomial-time LP solver. For a point $p \in \mathbb{R}^2$ and $i \in \{1, 2\}$, define $\delta_{p,i}$ as the sum of $x^{*}_{s}$ for all $s\in \mathcal{S}_{i}$ satisfying $p \in s$. Then we assign each point $p \in P$ to $P_{i}$, where $i \in \{1, 2\}$ is the index that maximizes $\delta_{p,i}$. 

Now we solve the two line instances $(P_i, P', \mathcal{S}_{i})$, for $i\in \{1, 2\}$ using Algorithm \ref{algo:exact_algo_hline}. Finally, we output the union of the solutions of these two line instances. We discard redundant squares from the solution, if any. For $i\in \{1, 2\}$, denote by $\mathcal{S}_i^{*}$ the solution returned by Algorithm \ref{algo:exact_algo_hline} for the line instance $(P_i, P', \mathcal{S}_{i})$. We state and prove the following useful lemma. 

\begin{lemma}\label{lemma:2approx_line}
For every $i\in \{1, 2\}$, $\mathsf{memb}(P', \mathcal{S}^{*}_{i})\leq 4\cdot OPT + 9$, where $\mathcal{S}^{*}_{i}$ is the solution returned by Algorithm \ref{algo:exact_algo_hline} for the instance $(P_i, P', \mathcal{S}_{i})$ and $OPT$ is the minimum membership for the slab instance.
\end{lemma}
\begin{proof}
We consider the case when $i=2$, i.e., all the squares intersect the top boundary line $L_2$ of the slab $\alpha$. The argument for the case of $i=1$ is identical and is not duplicated. Consider the leftmost maximum clique $Q$ in the intersection graph of the squares in $\mathcal{S}^{*}_{2}$. Suppose, $\mathsf{memb}(P', \mathcal{S}^{*}_{2})=k$ and the squares in $Q$ from left to right are $s_1, \ldots, s_k$. 

If $k\leq 13$, then $\mathsf{memb}(P', \mathcal{S}^{*}_{i})\leq 4\cdot OPT + 9$ is satisfied trivially since $OPT\geq 1$.

Assume that $k>13$. Denote by $p_{j}$ the bottom-most exclusive point of $s_{j}$. Let $p$ be any point in $P'$ contained in the ply region of $Q$. By Lemma \ref{lemma:domination}, for each $j\in J$, every input square containing $p_{j}$ also contains $p$. By Lemma \ref{lemma:domination2}, no input square $s$ contains $p_j, p_{j+1}, p_{j+2}$ for $j\in J$. Thus we can write
\begin{equation*}
\sum_{s\in \mathcal{S}_{2}, p\in s} x_{s} \geq \frac{1}{2} \sum_{\forall j\in J}\sum_{s\in \mathcal{S}_{2}, p_j\in s} x_{s}
\end{equation*}
Since a variable $x_{s}$ may appear at most twice in the double-sum on the right-hand side, we have multiplied by the factor $1/2$. The left-hand side of the above inequality is bounded above by the LP optimal $y^{*}$. Since every $p_{j}$ belongs to $P_2$, we have $\sum_{s\in \mathcal{S}_2, p_j\in s} x_{s} \geq 1/2$ from the partitioning criteria of $P$ into $P_1$ and $P_2$. From the proofs of Lemma \ref{lemma:domination} and Lemma \ref{lemma:domination2}, we observe that $|J|\geq k-9$, irrespective of the type of clique $Q$. So we can write
\begin{align*}
y^{*} &\geq \frac{1}{2}\cdot(k-9)\cdot \frac{1}{2}    \\
\implies k &\leq 4y^{*} + 9
\end{align*}
By definition, the optimal ply value, $OPT$, is at least $y^{*}$. Therefore, $k \leq 4\cdot OPT + 9$.
\end{proof}

\begin{lemma}\label{lemma: 2approx_slab}
For $i\in \{1, 2\}$, let $\mathcal{S}^{*}_{i}$ be the solution for the line instance $(P_i, P', \mathcal{S}_{i})$ obtained via Algorithm \ref{algo:exact_algo_hline}. Then $\mathcal{S}^{*} = \mathcal{S}^{*}_{1}\cup \mathcal{S}^{*}_{2}$ is a feasible solution to the slab instance $(P,  P', \mathcal{S})$ with $\mathsf{memb}(P', \mathcal{S}^{*})\leq 8\cdot OPT + 18$.
\end{lemma}
\begin{proof}
Since $\mathcal{S} = \mathcal{S}_{1}\cup \mathcal{S}_{2}$, so $\mathcal{S}^{*}$ is a feasible set cover for the slab instance $(P, P', \mathcal{S})$. Consider an arbitrary point $p\in P'$. Using Lemma \ref{lemma:2approx_line} for $i\in \{1, 2\}$, we get that the number of unit squares in $\mathcal{S}_{i}^{*}$ containing $p$ is at most $4\cdot OPT + 9$. Thus, $\mathsf{memb}(P', \mathcal{S}^{*})\leq 8\cdot OPT + 18$.
\end{proof}
\subsection{Putting everything together}
\begin{theorem}
GMMGSC problem admits an algorithm that runs in $O(m^{2}\log m + m^{2}n)$ time, and computes a set cover whose membership is at most $16\cdot OPT + 36$, where $OPT$ denotes the minimum membership.
\end{theorem}
\begin{proof}
We divide the plane into unit-height horizontal slabs. For each non-empty slab $\alpha$, we partition the slab instance into two subproblems, namely the line instances corresponding to the boundary lines of $\alpha$ using the LP-relaxation technique described in Section \ref{sec:slab_instance}. We solve each line instance using Algorithm \ref{algo:exact_algo_hline}. Then we output the union of the solutions thus obtained while discarding redundant squares, if any. Consider any point $p\in P'$. Suppose $p$ lies within a slab $\alpha$ whose boundary lines are $L_1$ and $L_2$. Since the squares containing $p$ must intersect either $L_1$ or $L_2$, $p$ can be contained in squares from at most $4$ subproblems. One is a line instance corresponding to $L_1$ where the points lie above $L_1$. The other is a line instance corresponding to $L_1$ where the points lie below $L_1$. The other two subproblems correspond to $L_2$. Thus, due to Lemma \ref{lemma:2approx_line}, the number of squares of our solution containing $p$ is at most $16\cdot OPT + 36$, where $OPT$ is the optimal membership value for the instance $(P, P', \mathcal{S})$. 

Formulating and solving the LP takes $\tilde{O}(nm)$ time \cite{LP_near_linear_Allen_Zhu2019}. The overall running time of the algorithm is dominated by the running time of Algorithm \ref{algo:exact_algo_hline}. There are at most $O(\min(n, m))$ line instances to solve. By a standard trick, the total running time remains $O(m^{2}\log m + m^{2}n)$ using Theorem \ref{thm:line_run_time}.
\end{proof}
\bibliographystyle{splncs04}
\bibliography{refs}

\begin{thebibliography}{10}
\providecommand{\url}[1]{\texttt{#1}}
\providecommand{\urlprefix}{URL }
\providecommand{\doi}[1]{https://doi.org/#1}

\bibitem{GeomSetCover_agarwal_pan_2014}
Agarwal, P.K., Pan, J.: Near-linear algorithms for geometric hitting sets and
  set covers. In: Proceedings of the Thirtieth Annual Symposium on
  Computational Geometry. p. 271–279. SOCG'14, Association for Computing
  Machinery, New York, NY, USA (2014). \doi{10.1145/2582112.2582152}

\bibitem{LP_near_linear_Allen_Zhu2019}
Allen-Zhu, Z., Orecchia, L.: Nearly linear-time packing and covering lp
  solvers. Mathematical Programming  \textbf{175}(1),  307--353 (May 2019).
  \doi{10.1007/s10107-018-1244-x}

\bibitem{socg23_minply}
Bandyapadhyay, S., Lochet, W., Saurabh, S., Xue, J.: {Minimum-Membership
  Geometric Set Cover, Revisited}. In: 39th International Symposium on
  Computational Geometry (SoCG 2023). Leibniz International Proceedings in
  Informatics (LIPIcs), vol.~258, pp. 11:1--11:14 (2023).
  \doi{10.4230/LIPIcs.SoCG.2023.11}

\bibitem{Biedl_Minply}
Biedl, T., Biniaz, A., Lubiw, A.: Minimum ply covering of points with disks and
  squares. Computational Geometry  \textbf{94},  101712 (2021).
  \doi{10.1016/j.comgeo.2020.101712}

\bibitem{Chan_Grant_pack_cover}
Chan, T.M., Grant, E.: Exact algorithms and apx-hardness results for geometric
  packing and covering problems. Computational Geometry  \textbf{47}(2, Part
  A),  112--124 (2014). \doi{10.1016/j.comgeo.2012.04.001}, special Issue: 23rd
  Canadian Conference on Computational Geometry (CCCG11)

\bibitem{CHAN_red_blue}
Chan, T.M., Hu, N.: Geometric red–blue set cover for unit squares and related
  problems. Computational Geometry  \textbf{48}(5),  380--385 (2015).
  \doi{10.1016/j.comgeo.2014.12.005}, special Issue on the 25th Canadian
  Conference on Computational Geometry (CCCG)

\bibitem{Clarkson_Varadarajan_geomsetcover}
Clarkson, K.L., Varadarajan, K.: Improved approximation algorithms for
  geometric set cover. Discrete {\&} Computational Geometry  \textbf{37}(1),
  43--58 (Jan 2007). \doi{10.1007/s00454-006-1273-8}

\bibitem{Durocher_Mondal_Minply}
Durocher, S., Keil, J.M., Mondal, D.: Minimum ply covering of points with unit
  squares. In: WALCOM: Algorithms and Computation: 17th International
  Conference and Workshops, WALCOM 2023, Hsinchu, Taiwan, March 22–24, 2023,
  Proceedings. p. 23–35. Springer-Verlag, Berlin, Heidelberg (2023).
  \doi{10.1007/978-3-031-27051-2\_3}

\bibitem{erlebach_soda_minply}
Erlebach, T., van Leeuwen, E.J.: Approximating geometric coverage problems. In:
  Proceedings of the Nineteenth Annual ACM-SIAM Symposium on Discrete
  Algorithms. p. 1267–1276. SODA '08, Society for Industrial and Applied
  Mathematics, USA (2008)

\bibitem{feige_1998}
Feige, U.: A threshold of ln n for approximating set cover. J. ACM
  \textbf{45}(4),  634–652 (jul 1998). \doi{10.1145/285055.285059}

\bibitem{geom_set_cover_2D_hardness}
Fowler, R.J., Paterson, M.S., Tanimoto, S.L.: Optimal packing and covering in
  the plane are np-complete. Information Processing Letters  \textbf{12}(3),
  133--137 (1981). \doi{10.1016/0020-0190(81)90111-3}

\bibitem{Hochbaum_Maass_covering}
Hochbaum, D.S., Maass, W.: Approximation schemes for covering and packing
  problems in image processing and vlsi. J. ACM  \textbf{32}(1),  130–136
  (jan 1985). \doi{10.1145/2455.214106}

\bibitem{takehiro_unique_coverage}
Ito, T., ichi Nakano, S., Okamoto, Y., Otachi, Y., Uehara, R., Uno, T., Uno,
  Y.: A polynomial-time approximation scheme for the geometric unique coverage
  problem on unit squares. Computational Geometry  \textbf{51},  25--39 (2016).
  \doi{10.1016/j.comgeo.2015.10.004}

\bibitem{kuhn_mmsc}
Kuhn, F., von Rickenbach, P., Wattenhofer, R., Welzl, E., Zollinger, A.:
  Interference in cellular networks: The minimum membership set cover problem.
  In: Wang, L. (ed.) Computing and Combinatorics. pp. 188--198. Springer Berlin
  Heidelberg, Berlin, Heidelberg (2005)

\bibitem{raman_geom_setcover_focs2014}
Mustafa, N.H., Raman, R., Ray, S.: Settling the apx-hardness status for
  geometric set cover. In: 55th {IEEE} Annual Symposium on Foundations of
  Computer Science, {FOCS} 2014, Philadelphia, PA, USA, October 18-21, 2014.
  pp. 541--550. {IEEE} Computer Society (2014). \doi{10.1109/FOCS.2014.64}

\bibitem{Mustafa_Ray}
Mustafa, N.H., Ray, S.: Improved results on geometric hitting set problems.
  Discrete {\&} Computational Geometry  \textbf{44}(4),  883--895 (Dec 2010).
  \doi{10.1007/s00454-010-9285-9}

\bibitem{Raz_setcover_hard_97}
Raz, R., Safra, S.: A sub-constant error-probability low-degree test, and a
  sub-constant error-probability pcp characterization of np. In: Proceedings of
  the Twenty-Ninth Annual ACM Symposium on Theory of Computing. p. 475–484.
  STOC '97, Association for Computing Machinery, New York, NY, USA (1997).
  \doi{10.1145/258533.258641}

\end{thebibliography}
\end{document}